\documentclass[12pt]{amsart}
\pdfoutput=1

\usepackage[utf8x]{inputenc}
\usepackage[T1]{fontenc}
\usepackage[draft=false,kerning=true]{microtype}

\usepackage{cyrillic}
\usepackage{colonequals}
\usepackage{mathtools}
\usepackage{hyperref}
\usepackage{fullpage}
\usepackage{amsmath,amsthm,amssymb}

\DeclarePairedDelimiter{\abs}{\lvert}{\rvert}

\DeclarePairedDelimiterXPP{\f}[2]{\foperator{#1}}(){}{#2}
\DeclarePairedDelimiter{\floor}{\lfloor}{\rfloor}
\newcommand{\foperator}[1]{\mathop{{#1}\empty{}}}

\newcommand{\N}{\mathbb{N}}
\DeclarePairedDelimiter{\norm}{\lVert}{\rVert}
\DeclarePairedDelimiterXPP{\oh}[1]{o}(){}{#1}
\DeclarePairedDelimiterXPP{\Oh}[1]{O}(){}{#1}
\newcommand{\Q}{\mathbb{Q}}
\newcommand{\R}{\mathbb{R}}
\renewcommand{\S}{\mathbb{S}}
\DeclarePairedDelimiter{\set}{\{}{\}}
\DeclarePairedDelimiterX{\setm}[2]{\{}{\}}{#1 \colon \mathopen{}#2}
\newcommand{\Z}{\mathbb{Z}}
\newcommand{\Alg}{\mathbb{A}}

\newtheorem{theorem}{Theorem}

\newtheorem{corollary}[theorem]{Corollary}

\newtheorem*{theorem*}{Theorem}
\newtheorem*{claim*}{Claim}

\newtheorem{lemma}{Lemma}[section]
\newtheorem{proposition}[lemma]{Proposition}

\theoremstyle{remark}
\newtheorem{remark}[lemma]{Remark}

\numberwithin{equation}{section}

\usepackage{todonotes}

\title{Decidability and $k$-Regular Sequences}

\author{Daniel Krenn}
\address{Fachbereich Mathematik,
Paris Lodron University of Salzburg,
Hellbrunnerstra{\ss}e 34,
5020 Salzburg,
Austria}
\email{math@danielkrenn.at}

\thanks{Daniel Krenn is supported by the Austrian Science Fund (FWF): P\,28466-N35.}

\author{Jeffrey Shallit}
\address{
School of Computer Science,
University of Waterloo,
Waterloo, ON  N2L 3G1,
Canada}
\email{shallit@uwaterloo.ca}

\thanks{This work was partially carried out by the authors during the workshop ``Numeration and Substitution 2019''
at the
Erwin Schr\"odinger International Institute for Mathematics and Physics (ESI), and they were supported by it.}

\thanks{This paper is published under a
  Creative Commons Attribution 4.0 International License.
  \raisebox{-0.5ex}{\includegraphics[height=1em]{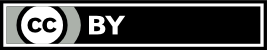}}
  (\url{http://creativecommons.org/licenses/by/4.0/})
  Copyright held by the authors.}
\begin{document}

\maketitle

\begin{abstract}
  In this paper we consider a number of natural decision problems
  involving $k$-regular sequences.
  Specifically, they arise from considering
  \begin{itemize}
  \item lower and upper bounds on growth rate;
  in particular  boundedness,
  \item images,
  \item regularity (recognizability by a deterministic finite automaton) of
  preimages, and
  \item factors, such as squares and palindromes,
  \end{itemize}
  of such sequences.
  We show that these decision problems are undecidable.
\end{abstract}

\section{Introduction}

A sequence $(a(n))_{n \geq 0}$ over a finite alphabet
is said to be {\it $k$-automatic}, for $k \geq 2$ an integer,
if its $k$-kernel
$$ K_k({(a(n))_{n \geq 0}})  = \setm{ (a(k^e n + i))_{n \geq 0} }{ e \geq 0,\ 0 \leq i < k^e }$$
is of finite cardinality.  There are many different equivalent definitions of this class of sequences 
\cite{Allouche-Shallit:1992:regular-sequences}.
It is well-known that many questions about these sequences, such as the growth rate
of $\sum_{0 \leq n < N} a(n)$, are decidable \cite{Cobham:1972}.
\vskip .1in

The so-called {\it $k$-regular sequences\/} form a natural
generalization of the automatic sequences.  These are sequences
$(a(n))_{n \geq 0}$ where the kernel $K_k ((a(n))_{n \geq 0}) $ is contained in a finitely
generated module.  Unlike the case of $k$-automatic sequences, it is
known that some decision problems involving $k$-regular sequences are
recursively unsolvable \cite{Allouche&Shallit:2003}.

In this paper we examine a number of natural decision problems involving $k$-regular sequences, and show that they are undecidable.

Some general results have recently been obtained by Honkala~\cite{Honkala:2021:quasi-universal-regular-sequences}.

\subsection{Recursively solvable decision problems}

A decision problem is one with a yes/no answer.
To say that a decision problem is solvable means
there exists an algorithm (or Turing machine) that
unerringly solves it on all inputs.  
Throughout this paper we use the terms ``recursively solvable'',
``solvable'', and ``decidable'' interchangeably, and similarly for the
terms ``recursively unsolvable'', ``unsolvable'', and ``undecidable''.

\subsection{Notation}

We let $\N_0$ denote the nonnegative integers (natural numbers) and $\N$ denote the
positive integers.

For a word $z$ with symbols chosen from a finite set~$D$, we let 
$\abs{z}$ denote its length and $\abs{z}_d$ the number of occurrences of
the letter~$d \in D$ in $z$.

For a fixed integer~$k \geq 2$, we consider base-$k$
representations with the usual digit set $D=\set{0,1,\dots,k-1}$.
For a nonnegative integer~$n$, we write $(n)_k$ for the standard
$k$-ary representation of~$n$, most significant digit first, and having no leading zeroes.
The representation of $0$ is the empty word.
Note that $(n)_k$ is a word over~$D$ and
that $\abs{(n)_k}=\floor{\log_k n}+1$ for $n > 0$.

\subsection{Hilbert's tenth problem}

Showing that a certain decision problem is recursively unsolvable is often carried out
by constructing a (Turing) reduction from another decision problem 
already known to be recursively unsolvable. One such problem is
Hilbert's tenth problem \cite{Davis-Putnam-Robinson:1961:exp-diophantine-equations,
  Matijasevic:1970:diophantiness-enumerable-sets}:
\begin{theorem*}[Hilbert's tenth problem; variant]
  The decision problem
  \begin{quote}
    ``Given a multivariate polynomial $p$ with integer coefficients,
    do there exist natural numbers $x_1$, $x_2$, \dots, $x_t$ such that $p(x_1,
    \dots, x_t) = 0$?''
  \end{quote}
  is recursively unsolvable.
\end{theorem*}
The analogous problem, where the
$x_i$ need to be positive, is also recursively unsolvable.
We will reduce from this problem quite frequently, namely in
Theorems~\ref{thm:unsolvable:Omega-n},
\ref{thm:unsolvable:preimage-dfa}, \ref{thm:unsolvable:image-N},
\ref{thm:unsolvable:image-Z} and \ref{thm:unsolvable:value-twice}.

\subsection{Representation of \texorpdfstring{$k$}{k}-regular sequences}
A $k$-regular sequence $(f(n))_{n \geq 0}$ can be finitely represented in a number of different ways, of which two are the most useful.  First, a set of identities in terms of sequences in the $k$-kernel, where each identity represents a subsequence $(f(k^e n + i))_{n \geq 0}$, $0 \leq i < k^e$, as a linear combination of subsequences in the $k$-kernel, and a set of initial values.  Together it must be possible to compute $f(n)$ for all $n$ from this set of identities and initial values. 

Second, a {\it linear representation\/} for $(f(n))_{n\geq 0}$, which consists of a $1 \times r$ row vector $v$, an $r \times 1$ column vector $w$,
and $k$ square matrices $M_0, M_1, \ldots, M_{k-1}$ of dimension
$r \times r$ such that
$$ f(n) = v M_{n_0} \cdots M_{n_{s-1}} w,$$
for all $n$, where $(n)_k = n_{s-1} \cdots n_0$.  Of course, the empty product of matrices is the identity matrix.  The \textit{rank} of a linear representation is $r$.   
See \cite[Theorem~2.2]{Allouche-Shallit:1992:regular-sequences}. 
A linear representation is \textit{minimal} if it has smallest possible
rank for the corresponding sequence.

For example, consider the $2$-regular sequence $(s_2(n))_{n\ge0}$, which counts the number of $1$'s
in the binary representation of $n$.  Then it is easy to see that
\begin{align*}
s_2 (0) &= 0 \\
s_2 (2n) &= s_2 (n) \\
s_2 (4n+1) &= s_2 (2n+1) \\
s_2 (4n+3) &= -s_2 (n) + 2 s_2 (2n+1)
\end{align*}
is an example of the former
representation, and
$$
v  = [ 0 \ \ 1 ]; \quad
M_0 = \left[ \begin{array}{cc}
    1 & 0 \\
    0 & 1 
    \end{array} \right]; \quad
M_1 = \left[ \begin{array}{cc}
    0 & -1 \\
    1 & 2
    \end{array} \right]; \quad
w = \left[ \begin{array}{c}
    1 \\
    0 \end{array} \right] 
$$
is an example of the latter.

From now on, when we say an algorithm is ``given'' a $k$-regular sequence as input, we actually mean that the input is either one of these two representations.
Note that we can transform between these two representations effectively, that is, with an algorithm \cite{Berstel-Reutenauer:2011:noncommutative-rational-series}.

Some of our theorems involve algebraic numbers.   When we say we are ``given'' a real algebraic number $\alpha$, we mean we are given the minimal polynomial for $\alpha$, together with a rational interval that contains $\alpha$ and none of its conjugates.  As is well-known \cite{Froehlich-Shepherdson:1956:effective-proc-field-theory, Griffor:1999:handbook-computability-theory},
we can effectively carry out arithmetic on algebraic numbers represented in this way.

\subsection{Closure properties of \texorpdfstring{$k$}{k}-regular sequences}

In this section we recall some closure properties of $k$-regular
sequences: that is, which operations on sequences preserve the property of
$k$-regularity.  For more details, see
\cite{Allouche-Shallit:1992:regular-sequences}.  It is important to
note that not only do these operations preserve $k$-regularity; they
also are {\it effectively\/} $k$-regular.  Let $\circ$ be some operation
on sequences.  By the operation $\circ$ being ``effectively
$k$-regular'', we mean that there is an algorithm that, given
some representation of $k$-regular sequences
${\bf a} = (a(n))_{n \geq 0}$ and ${\bf b} = (b(n))_{n \geq 0}$,
computes a representation for ${\bf a} \circ {\bf b}$.

\begin{theorem*}
The class of $k$-regular sequences is (effectively) closed under the following operations:
\begin{itemize}
    \item[(a)] sum, ${\bf a}+{\bf b} = (a(n) + b(n))_{n \geq 0}$;
    \item[(b)] product,
    ${\bf a} {\bf b} = (a(n) \, b(n))_{n \geq 0}$;
    \item[(c)] convolution,
    ${\bf a} \star {\bf b} =
    (\sum_{0 \leq i \leq n} a(i) \, b(n-i) )_{n \geq 0}$;
    \item[(d)] perfect shuffle,
    ${\bf a} \,\sha\, {\bf b} =
    {\bf c} = (c(n))_{n \geq 0}$,
    where $c(2i) = a(i)$ and
    $c(2i+1) = b(i)$ for $i \geq 0$.   The same is true for
    $t$-way perfect shuffle, where we combine $t$ sequences analogously.
\end{itemize}
\end{theorem*}
For proofs, see  \cite{Allouche-Shallit:1992:regular-sequences}.

\begin{remark}\label{rem:poly-regular}
  Let $p$ be a multivariate polynomial with integer coefficients, and suppose
  $d_1, \ldots, d_t\in\set{0,\dots,k-1}$. Then
  \begin{equation*}
    (\f{p}{\abs{z}_{d_1}, \abs{z}_{d_2}, \dots, \abs{z}_{d_t}})_{n \geq 0}
  \end{equation*}
  with $z = (n)_k$, is (effectively) a $k$-regular sequence in~$n\in\N_0$ over~$\Z$. This is
  true because the number of occurrences $\abs{z}_d$ of a digit~$d$
  in the standard $k$-ary expansion~$z=(n)_k$ is a $k$-regular
  sequence by
  \cite[Theorem~6.1]{Allouche-Shallit:1992:regular-sequences} and, as above,
  $k$-regular sequences are closed under term-by-term addition and multiplication.
\end{remark}

\section{Equality of \texorpdfstring{$k$}{k}-regular sequences}

Before we consider growth of $k$-regular sequences in the next section,
we study decidability of whether one $k$-regular sequence equals another.

\begin{theorem}
  Let $k\geq2$ be an integer and $\S\supseteq\N_0$ be a set.
  The decision problem
  \begin{quote}
    ``Given two $k$-regular sequences $(f(n))_{n \geq 0}$ and $(g(n))_{n \geq 0}$ over~$\S$,
    does $f(n) = g(n)$ for all $n$?''
  \end{quote}
  is recursively solvable.
\end{theorem}

\begin{proof}
  The sequence~$(f(n)-g(n))_{n \geq 0}$ is $k$-regular by the closure properties of $k$-regular sequences.
  We compute a linear representation of this sequence and apply a minimization algorithm
  (see Berstel and Reutenauer~\cite{Berstel-Reutenauer:2011:noncommutative-rational-series}) to it.
  This results in a linear representation of minimum rank. This rank is $0$ iff we started with the zero sequence, iff
  $f(n) = g(n)$ for all $n$.
\end{proof}

\section{Growth of \texorpdfstring{$k$}{k}-regular sequences}

We use the standard notation for asymptotic growth of sequences.
Let $(f(n))_{n\ge0}$ and $(g(n))_{n\ge0}$ be sequences. Then  
\begin{itemize}
\item $f(n) \in O(g(n))$ means that there exist $n_0$, $c > 0$ such that
  $f(n) \leq c g(n)$ for all $n \geq n_0$, and
\item $f(n) \in \Omega(g(n))$ means that there exist $n_0$, $c > 0$ such that
  $f(n) \geq c g(n)$ for all $n \geq n_0$.
\end{itemize}
For simplicity, we sometimes say that the sequence $(f(n))_{n\ge0}$ is in $O(g(n))$ or  $\Omega(g(n))$.

In what follows, $\Alg$ denotes the set of real algebraic numbers.

\subsection{Lower bounds}
\label{sec:lower-bound}

\begin{theorem}\label{thm:unsolvable:Omega-n}
  Let $k\geq2$ be an integer and $\S$ be a set with
  $\N\subseteq\S\subseteq\Alg$. The decision problem
  \begin{quote}
    ``Given a $k$-regular sequence $(f(n))_{n \geq 0}$ over $\S$,
    is $f(n)$ in $\f{\Omega}{n}$?''
  \end{quote}
  is recursively unsolvable.
\end{theorem}

\begin{proof}
  We reduce from Hilbert's tenth problem.
  For a given multivariate polynomial~$p$ in $t$~variables over~$\Z$,
  we choose $r\in\N$ such that $K=k^r > t+1$, and 
  we construct the sequence
  \begin{equation*}
    f(n) \coloneqq
    (n+1)
    \bigl(\f{p}{\abs{z}_1, \abs{z}_2, \dots, \abs{z}_t}\bigr)^2 \,
    (\abs{z}_{t+1} + 1)
  \end{equation*}
  with $z = (n)_K$.
  The sequence $(f(n))_{n \geq 0}$ is $K$-regular (see Remark~\ref{rem:poly-regular})
  and therefore $k$-regular as well
  \cite[Theorem 2.9]{Allouche-Shallit:1992:regular-sequences}.

  The following claim shows that the above indeed provides a
  reduction.
  \begin{claim*}
    The sequence $(f(n))_{n \geq 0}$ is not in $\Omega(n)$ iff there exist
    nonnegative integers $x_1$, $x_2$, \ldots, $x_t$ such that
    $p(x_1, \dots, x_t) = 0$.
  \end{claim*}
  To see this, note that 
  $f(n) = 0$ iff at least one factor is zero, but the first
  and third factors defining $f$ are never zero.  Hence $f(n) = 0$ iff
  $\f{p}{\abs{z}_1, \abs{z}_2, \dots, \abs{z}_t} = 0$.
  Moreover, note that if a zero of $(f(n))_{n\ge0}$ occurs once, it
  occurs infinitely often by the third factor of the definition for $f$.

  Thus, if $p(x_1, \dots, x_t) = 0$ has at least one solution, then $f(n) = 0$ for 
  infinitely many $n$, and hence $f(n)$ is not in $\Omega(n)$.
  Otherwise, if $p(x_1, \dots, x_t) = 0$ does not have any solution, then
  its absolute value is at least $1$, so the value of $f(n)$
  is at least $n+1$, and hence $f(n)$ is in $\Omega(n)$. This 
  completes the proof of the claim and consequently the proof of Theorem~\ref{thm:unsolvable:Omega-n}.
\end{proof}

Theorem~\ref{thm:unsolvable:Omega-n} can be extended to other growth rates as well.  

\begin{corollary}\label{cor:unsolvable:Omega}
  Let $k\geq2$ be an integer. Suppose $\sigma$ is a real number
  and $\ell$ is a nonnegative integer, not both zero. Assume
  that $k^\sigma$ is an algebraic number.   Let $\S$ be a ring with
  $\Z\subseteq\S\subseteq\Alg$ and containing $k^\sigma$.
  The decision problem
  \begin{quote}
    ``Given a $k$-regular sequence $(f(n))_{n \geq 0}$ over~$\S$,
    is $f(n)$ in $\f[\big]{\Omega}{n^\sigma (\log n)^\ell}$?''
  \end{quote}
  is recursively unsolvable.
\end{corollary}

\begin{remark}\label{rem:growth-theta-sequence}
  For a real number~$\sigma$ and a nonnegative integer~$\ell$ we
 construct a $k$-regular sequence $(h_{\sigma,\ell}(n))_{n \geq 0}$ with positive
  terms (except for the first few terms, which may be $0$) satisfying
  \begin{equation*}
    h_{\sigma,\ell}(n) \in \f[big]{\Theta}{n^\sigma (\log n)^\ell},
  \end{equation*}
  as follows.

  Set $H_0=\cdots=H_{k-1}=J_{\ell+1}(k^\sigma)$, where
  $J_{\ell+1}(k^\sigma)$ is a Jordan block, of size $\ell + 1$, corresponding to the eigenvalue $k^\sigma$. We set
  \begin{equation*}
    h_{\sigma,\ell}(n) = e_1 \f{H}{n}\, e_{\ell+1},
  \end{equation*}
  where $H(n) = H_{n_0} \cdots H_{n_{\ell-1}}$ for $(n)_k = n_{\ell-1} \cdots n_0$,
  and the $e_i$ are the $i$th unit vectors. Therefore, the sequence~$h_{\sigma,\ell}(n)$
  is $k$-regular, as it is defined by a linear representation.
  Explicitly, we have
  \begin{equation*}
    h_{\sigma,\ell}(n) = \binom{s}{\ell} k^{(s-\ell)\sigma},
  \end{equation*}
  where $s=\floor{\log_k n}+1$. Thus, this sequence's asymptotic behavior is
  $\f[big]{\Theta}{n^\sigma (\log n)^\ell}$. If $\ell=0$, then no term is $0$.   If $\ell \not=0$, then
  only terms with $n \leq k^{\ell-1}$ are~$0$.
\end{remark}

\begin{proof}[Proof of Corollary~\ref{cor:unsolvable:Omega}]
  The proof runs along the same lines as the proof of
  Theorem~\ref{thm:unsolvable:Omega-n}.  
  For a given multivariate polynomial~$p$ in $t$~variables over~$\Z$,
  we instead choose $r\in\N$ such that $K=k^r > t+1$, and 
  we define
  \begin{equation*}
    f(n) \coloneqq
    (\f{h_{\sigma,\ell}}{n} + 1)
    \bigl(\f{p}{\abs{z}_1, \abs{z}_2, \dots, \abs{z}_t}\bigr)^2 \,
    (\abs{z}_{t+1} + 1)
  \end{equation*}
  with $z = (n)_K$ and $(\f{h_{\sigma,\ell}}{n})_{n\ge0}$ of
  Remark~\ref{rem:growth-theta-sequence}. Note that the factor
  $(\f{h_{\sigma,\ell}}{n} + 1)$ of $f(n)$ is increasing and always positive.
\end{proof}

\subsection{Upper bounds}
\label{sec:upper-bound}

Let $(h(n))_{n \geq 0}$ be a sequence. We say that a sequence~$(M(n))_{n\ge0}$ of matrices
with entries in a set~$\S\subseteq\Alg$ is in $\Oh{h(n)}$,
formally written as usual as
\begin{equation*}
  M(n) \in \Oh{h(n)},
\end{equation*}
if each sequence of a fixed entry (fixed row and column) of the matrices is in
$\Oh{h(n)}$. Rephrased, this means that the sequence of maximum
norms of the matrices lies in $\Oh{h(n)}$. By the equivalence of
norms, this is also true for any other norm.
As in the one-dimensional case, we say that the sequence~$(M(n))_{n\ge0}$ is
\emph{bounded} if it lies in~$\Oh{1}$.

\begin{remark}\label{rem:growth:linear-deviation}
  Let $\sigma\in\R$ and $\ell\in\N_0$, and set $h(n) = n^\sigma (\log
  n)^\ell$. If $m(n) = cn + \tau(n)$ for some constant~$c\neq0$ and
  some sequence~$\tau(n)\in\oh{n}$, then
  \begin{equation*}
    \Oh{h(m)} = \Oh{h(n)}.
  \end{equation*}
  as $n\to\infty$.
  This follows from
  \begin{align*}
    h(m(n)) &= \bigl(cn + \tau(n)\bigr)^\sigma \bigl(\log (cn + \tau(n))\bigr)^\ell \\
    &= c^\sigma h(n)
    \Bigl(1 + \frac{\tau(n)}{cn}\Bigr)^\sigma
    \biggl(1 + \frac{\log\bigl(c+\frac{\tau(n)}{n}\bigr)}{\log n}\biggr)^\ell \\
    &= c^\sigma h(n) \bigl(1+\oh{1}\bigr).
  \end{align*}
\end{remark}

\begin{theorem}\label{thm:unsolvable:bounded}
  Let $k\geq2$ be an integer and $\S$ be a ring with
  $\Q\subseteq\S\subseteq\Alg$. The decision problem
  \begin{quote}
    ``Given a $k$-regular sequence $(f(n))_{n \geq 0}$ over~$\S$,
    is $f(n)$ bounded?''
  \end{quote}
  is recursively unsolvable.
\end{theorem}

\begin{remark}
  The above problem is decidable for $k$-regular sequences that have
  a linear representation with integer matrices;  
  see the algorithm of
  Mandel and Simon~\cite{Mandel-Simon:1977:finite-matrix-semigroups}
  for matrices with nonnegative entries and the algorithm of
  Jacob~\cite{Jacob:1977:decidablity-matrix-semigroups,
    Jacob:1977:algorithm-matrix-semigroups, 
    Jacob:1978:finiteness-matrix-semigroups}
  for general integer matrices.
\end{remark}

  Given square matrices~$F_0$, \ldots, $F_{k-1}$ over a ring~$\S$
  all of the same dimension, 
  we set $F(n) =
  F_{n_0} \cdots F_{n_{s-1}}$ for $(n)_k = n_{s-1} \cdots n_0$, and we call
  the sequence~$(F(n))_{n\ge0}$ the
  {\it matrix-valued linear representation sequence\/} of the set $\set{F_0, \ldots, F_{k-1}}$.
  We use this notion in the following lemma, which extends decidability
  from $k$-regular sequences to matrix-valued linear representation sequences.

\begin{lemma}\label{lem:property:k-regular-products}
  Let $\S$ be a ring.
  Let $P$ be a property of a sequence over~$\S$,
  i.e., $P$ for each sequence over~$\S$ is
  either true or false. Suppose we can extend property~$P$ to
  sequences of matrices over~$\S$ in one of the following ways: Property~$P$
  holds for a sequence of matrices iff $P$ holds for
  \begin{enumerate}
  \item\label{itm:lem:property:k-regular-products:all} all sequences or
  \item\label{itm:lem:property:k-regular-products:any} any sequence
  \end{enumerate}
  consisting of a fixed entry (fixed row and column).

  If $P$ is recursively solvable for $k$-regular sequences over~$\S$,
  then $P$ is recursively
  solvable for matrix-valued linear representation sequences of a set of $k$ square
  matrices over~$\S$, all of the same dimension.
\end{lemma}

We will use this lemma in the proof of
Theorem~\ref{thm:unsolvable:bounded} with the property~$P$ being
the boundedness of a sequence and in the proof of
Theorem~\ref{thm:unsolvable:poly}, where $P$ is true iff a sequence
does not have at least polynomial growth.

\begin{proof}[Proof of Lemma~\ref{lem:property:k-regular-products}]
  We show (\ref{itm:lem:property:k-regular-products:all}). Then
  (\ref{itm:lem:property:k-regular-products:any}) follows by using the
  negation of property~$P$.

  We define
  \begin{equation*}
    f_{i,j}(n) = e_i F(n) e_j
  \end{equation*}
  with $F(n)$ as above and
  where $e_i$ is the $i$th unit vector. Therefore, $(f_{i,j}(n))_{n\ge0}$ is the
  sequence of entries in row~$i$ and column~$j$ of~$(F(n))_{n\ge0}$. All
  sequences~$(f_{i,j}(n))_{n\ge0}$ are $k$-regular, as they are defined by a
  linear representation.
  Clearly all these
  sequences~$(f_{i,j}(n))_{n\ge0}$ satisfy~$P$ iff $F(n)$ satisfies~$P$.

  As the question of deciding property $P$ of a $k$-regular sequence is
  recursively solvable, we can decide $P$ for
  each of the (finitely many) distinct sequences $(f_{i,j}(n))_{n \geq 0}$
  and therefore can decide~$P$ for~$(F(n))_{n\ge0}$.
\end{proof}

\begin{proof}[Proof of Theorem~\ref{thm:unsolvable:bounded}]
  We reduce from the question of
  boundedness of all products of matrices over the rationals, which is
  not recursively solvable; see Blondel and
  Tsitsiklis~\cite{Blondel-Tsitsiklis:2000:boundedness-undecidable}.

  For a given set of matrices $\set{F_0, \ldots, F_{k-1}}$, 
  for all $n_0$, \dots, $n_{s-1}\in\set{0,\dots,k-1}$,
  there is either a largest index~$j\in\set{1,\dots,s}$ with $n_{j-1}\neq0$
  or $j=0$ and such that $F_{n_0}\cdots F_{n_{s-1}} = F_{n_0}\cdots F_{n_{j-1}} F_0^{s-j}$.
  Whenever we are now deciding the boundedness of matrix products,
  we split as above and consider the factors~$F_0^{s-j}$ on the right-hand side
  and the remaining product separately.
  As for the $F_0^{s-j}$, we can decide the
  boundedness of powers of a single matrix from knowledge of its Jordan decomposition.
  And, as for the remaining factors, the statement of the theorem follows
  by using the reduction that is provided by
  Lemma~\ref{lem:property:k-regular-products} with property~$P$ being
  the boundedness of a sequence.
\end{proof}

\begin{corollary}\label{cor:unsolvable:Oh}
  Let $k\geq2$ be an integer, $\sigma$ a real number and $\ell$ a
  nonnegative integer.  Assume that $k^\sigma$ is
  an algebraic number.  Let $\S$ be a ring with
  $\Q\subseteq\S\subseteq\Alg$ and containing $k^\sigma$.
  The decision problem
  \begin{quote}
    ``Given a $k$-regular sequence $(f(n))_{n \geq 0}$ over~$\S$,
    is $f(n)$ in $\Oh{n^\sigma (\log n)^\ell}$?''
  \end{quote}
  is recursively unsolvable.
\end{corollary}

\begin{proof}
  We reduce from the decision problem stated in Theorem~\ref{thm:unsolvable:bounded}.
  For a $k$-regular
  sequence $(g(n))_{n\ge0}$, we construct $\f{f}{n} = \f{g}{n} \f{h_{\sigma,\ell}}{n}$
  with $\f{h_{\sigma,\ell}}{n}$ as defined in
  Remark~\ref{rem:growth-theta-sequence}.

  Then, the $k$-regular sequence~$(f(n))_{n \geq 0}$ is in $\Oh{n^\sigma (\log n)^\ell}$
  iff $g(n)$ is in $\Oh{1}$, i.e., bounded. Therefore deciding if a
  $k$-regular sequence is in $\Oh{n^\sigma (\log n)^\ell}$ implies deciding
  the boundedness of a $k$-regular sequence, which contradicts
  Theorem~\ref{thm:unsolvable:bounded}.
\end{proof}

Let $\sigma\in\R$ and $\ell\in\N_0$. We say that a sequence~$(f(n))_{n \geq 0}$ has {\it exact
growth\/} $n^\sigma (\log n)^\ell$ if
\begin{equation*}
  f(n) \in \Oh[\big]{n^\sigma (\log n)^\ell}
\end{equation*}
but for all $\sigma'\in\R$ and $\ell'\in\N_0$ with
 $(\sigma',\ell')$ lexicographically
 smaller than $(\sigma,\ell)$ we have
\begin{equation*}
  f(n) \not\in \Oh[\big]{n^{\sigma'} (\log n)^{\ell'}}.
\end{equation*}

\begin{proposition}\label{pro:growth-matrices-vs-sequence}
  Let $(f(n))_{n \geq 0}$ be a $k$-regular sequence over a field~$\S \subseteq \Alg$
  with matrices~$F_0, \dots,
  F_{k-1}$ of a minimal linear representation, and set $F(n) =
  F_{n_0} \cdots F_{n_{s-1}}$ for $(n)_k = n_{s-1} \cdots n_0$. Let
  $\sigma\in\R$ and $\ell\in\N_0$, and
  set $h(n) = n^\sigma (\log n)^\ell$. Then
  \begin{equation*}
    f(n) \in \Oh{h(n)}
  \end{equation*}
  if and only if
  \begin{equation*}
    F(n) \in \Oh{h(n)}.
  \end{equation*}
  In particular, both $f(n)$ and $F(n)$ have exactly the same growth rate.
\end{proposition}

\begin{proof}
  Let $\lambda$ and $\gamma$ be the vectors of our minimal
  representation, i.e., $f(n) = \lambda \f{F}{n} \gamma$ for all
  $n\in\N_0$. We start with the easy direction: As $f(n)$ is a finite
  linear combination of the entries in the matrix $F(n)$ and each of
  these entries is in $\Oh{h(n)}$, we have $f(n)$ is in 
  $\Oh{h(n)}$.

  Conversely, suppose $F(n)$ is not in
  $\Oh{h(n)}$. As there are only finitely
  many entries in each matrix~$F(n)$, we can assume that one entry of
  $F(n)$ is not in $\Oh{h(n)}$. Let $(g(n))_{n\ge0}$ denote the sequence
  of this fixed entry of the matrices.

  Now, as our linear representation is minimal, there exist finite
  subsets~$P$, $Q \subseteq \N_0$ and
  coefficients~$c_p$, $d_q \in \S\setminus\set{0}$ for $p\in P$, $q\in Q$
  such that
  \begin{equation}\label{eq:growth-matrices-vs-sequence:sum}
    g(n) = \lambda \biggl(\, \sum_{p\in P}\sum_{q\in Q} c_p d_q \f{F}{p} \f{F}{n} \f{F}{q} \biggr) \gamma
  \end{equation}
  for all $n\in\N_0$;
  see~\cite[Corollary 2.3]{Berstel-Reutenauer:2011:noncommutative-rational-series}.
  As $g(n)$ is not in $\Oh{h(n)}$, one of the finitely many
  summands
  \begin{equation*}
    c_p d_q \cdot \lambda \f{F}{p} \f{F}{n} \f{F}{q} \gamma
  \end{equation*}
  of~\eqref{eq:growth-matrices-vs-sequence:sum}, where $p \in P$ and
  $q \in Q$, is not in $\Oh{h(n)}$. Dividing this summand by
  $c_p d_q$ yields a subsequence of $(f(n))_{n \geq 0}$, namely
  \begin{equation*}
    \lambda \f{F}{p} \f{F}{n} \f{F}{q} \gamma
    = \f{f}{m(n)}
  \end{equation*}
  with $m(n) = pk^{\abs{(n)_k}+\abs{(q)_k}} + nk^{\abs{(q)_k}} + q$ for
  all $n\in\N_0$.  As $\abs{(n)_k}=\floor{\log_k n}+1$, we have $m(n)
  = cn + o(n)$ for some constant~$c$, and therefore, by
  Remark~\ref{rem:growth:linear-deviation}, we obtain that the
  subsequence~$\f{f}{m(n)}$ is not in $\Oh{h(n)} =
  \Oh{h(m(n))}$. Thus the sequence~$(f(n))_{n \geq 0}$ itself is not in
  $\Oh{h(n)}$.
\end{proof}

\subsection{Polynomial growth}
\label{sec:polynomial-growth}

The growth of a $k$-regular sequence is always at most polynomial. To
be precise, for a $k$-regular sequence~$(f(n))_{n \geq 0}$ with values in~$\Alg$,
there exists a real constant~$\sigma\geq0$ such that $f(n) \in \Oh{n^\sigma}$; see
\cite[Theorem~2.10]{Allouche-Shallit:1992:regular-sequences}.

\begin{theorem}\label{thm:unsolvable:poly}
  Let $k\geq2$ be an integer and $\S$ be a ring with
  $\Q\subseteq\S\subseteq\Alg$. The decision problem
  \begin{quote}
    ``Given a $k$-regular sequence $(f(n))_{n \geq 0}$ over~$\S$,
    does $f(n)$ have at least polynomial growth,
    i.e., does there exist $\sigma > 0$ such that
    $f(n)$ is not in $\Oh{n^\sigma}$?''
  \end{quote}
  is recursively unsolvable.
\end{theorem}
Consequently, the decision problem in this theorem is whether a $k$-regular sequence has
polynomial growth or a smaller growth.

Recall that the \textit{joint spectral radius} of a finite set $S$
of square matrices is defined to be
$$\rho(S) = \lim_{\ell \rightarrow \infty} \max\setm{ \norm{A_{1} \cdots A_{\ell}}^{1/\ell}}{A_i \in S},$$
where $\norm{\,\cdot\,}$ is any matrix norm.

\begin{proposition}\label{pro:jsr-growth}
  Let $\S\subseteq\Alg$ be a ring, let $\rho\in\R$ be positive,
  let $F_0$, \dots, $F_{k-1}$ be square matrices over~$\S$ all of the same dimension, and set $F(n) =
  F_{n_0} \cdots F_{n_{s-1}}$ for $(n)_k =n_{s-1} \cdots n_0$.
  Then the following two statements are equivalent:
  \begin{enumerate}
  \item The joint spectral radius of $F_0$, \dots, $F_{k-1}$ is $\rho$.
  \item For all $\varepsilon>0$ we have $F(n) \in
    \Oh[\big]{n^{(\log_k\rho)+\varepsilon}}$ and
    $F(n) \not\in \Oh[\big]{n^{(\log_k\rho)-\varepsilon}}$ as $n\to\infty$,
    and we have
    $F_0^s \in \Oh[\big]{(k^s)^{(\log_k\rho)+\varepsilon}}$ as $s\to\infty$.
  \end{enumerate}
  In particular, the joint spectral radius~$\rho$ of $F_0$, \dots, $F_{k-1}$ is bounded by some positive
  $\rho'\in\R$, i.e., $\rho \leq \rho'$, iff for all $\varepsilon>0$ we
  have $F(n) \in \Oh[\big]{n^{(\log_k\rho')+\varepsilon}}$ as $n\to\infty$ and
  $F_0^s \in \Oh[\big]{(k^s)^{(\log_k\rho')+\varepsilon}}$ as $s\to\infty$.

  If $\S$ is a field and the matrices $F_0$, \dots, $F_{k-1}$ are of a minimal
  representation of a $k$-regular sequence~$(f(n))_{n \geq 0}$, then we may
  replace~$F(n)$ by~$f(n)$ in the statements of this proposition.
\end{proposition}

\begin{corollary}\label{cor:jsr-growth}
  Let the assumptions of Proposition~\ref{pro:jsr-growth} hold and suppose $F_0$ is
  the zero-matrix.
  Then the joint spectral radius~$\rho$ of $F_1$, \dots, $F_{k-1}$ is bounded by some positive
  $\rho'\in\R$, i.e., $\rho \leq \rho'$, iff for all $\varepsilon>0$ we
  have $F(n) \in \Oh[\big]{n^{(\log_k\rho')+\varepsilon}}$ as $n\to\infty$.
\end{corollary}

We will use this corollary with $\rho'=1$
in the proof of Theorem~\ref{thm:unsolvable:poly} to connect polynomial growth with
the joint spectral radius~$\rho$.

\begin{proof}[Proof of Proposition~\ref{pro:jsr-growth}]
  In this proof, we suppose that $s$ and $n$ are related by $s =
  \floor{\log_k n} + 1$. Then, by Remark~\ref{rem:growth:linear-deviation},
  we have $\Oh{n^{\sigma}} = \Oh{k^{s\sigma}}$ as
  $n\to\infty$ for any $\sigma$.

  We have that for any fixed real  $\sigma$,
  \begin{equation*}
    \norm{F(n)} \in \Oh{n^{\sigma}}
    = \Oh{k^{s\sigma}}
  \end{equation*}
  as $n\to\infty$ is equivalent to
  \begin{equation}\label{eq:jsr-growth:max}
    \max_{k^{s-1} \leq n < k^{s}} \norm{F(n)} \in \Oh{k^{s\sigma}}
  \end{equation}
  as $s\to\infty$, because~$s$ is the same for all~$n$ within
  the given range of the argument of the maximum and the right-hand
  side~$\Oh{k^{s\sigma}}$ only depends on~$s$ (and not
  on~$n$).

  We set
  \begin{equation*}
    \rho_s \coloneqq
    \max_{n_0,\dots,n_{s-1}\in\set{0,\dots,k-1}}
    \norm{F_{n_0}\cdots F_{n_{s-1}}}^{1/s}
  \end{equation*}
  Then the bound~\eqref{eq:jsr-growth:max} together with
  $F_0^s \in \Oh{k^{s\sigma}}$ is
  equivalent to
  \begin{equation}\label{eq:jsr-growth:rho-ell}
    k^{s\log_k\rho_s} = \rho_s^s \in \Oh{k^{s\sigma}}
  \end{equation}
  as $s\to\infty$,
  because there is a constant~$c>0$ (only depending on the used norm)
  such that for all $n_0$, \dots, $n_{s-1}\in\set{0,\dots,k-1}$,
  there is either a largest index~$j\in\set{1,\dots,s}$ with $n_{j-1}\neq0$
  or $j=0$ and
  \begin{equation*}
    \norm{F_{n_0}\cdots F_{n_{s-1}}}
    \leq c \norm{F_{n_0}\cdots F_{n_{j-1}}} \cdot \norm{F_0^{s-j}}
    \in \Oh{k^{j\sigma}} \, \Oh{k^{(s-j)\sigma}}
    = \Oh{k^{s\sigma}}.
  \end{equation*}

  Consequently, the bound~\eqref{eq:jsr-growth:rho-ell} is
  equivalent to the existence of an $S\in\N_0$
  such that for all $s \geq S$, the inequality $\log_k\rho_s
  \leq \sigma$ holds. So much for our preliminary considerations.

  Now let $\varepsilon>0$. Then $F(n) \not\in
  \Oh[\big]{n^{(\log_k\rho)-\varepsilon}}$, $F(n) \in
  \Oh[\big]{n^{(\log_k\rho)+\varepsilon}}$ and  $F_0^s \in
  \Oh[\big]{(k^s)^{(\log_k\rho)+\varepsilon}}$ iff there is an $S\in\N_0$ such
  that for all $s \geq S$, the inequalities
  \begin{equation*}
    (\log_k\rho)-\varepsilon
    < \log_k\rho_s
    \leq (\log_k\rho)+\varepsilon
  \end{equation*}
  hold. But this is equivalent to $\log_k\rho = \lim_{s\to\infty}
  \log_k\rho_s$ and therefore equivalent to
  \begin{equation*}
    \rho = \lim_{s\to\infty} \rho_s,
  \end{equation*}
  which completes the proof of the equivalence.

  If $f(n)$ is as in the proposition, then by
  Proposition~\ref{pro:growth-matrices-vs-sequence}
  we have the equivalence of
  \begin{equation*}
    f(n) \in \Oh{n^{\sigma}}.
  \end{equation*}
  to
  \begin{equation*}
    \norm{F(n)} \in \Oh{n^{\sigma}}
    = \Oh{k^{s\sigma}}
  \end{equation*}
  as $n\to\infty$ for any fixed real algebraic $\sigma$, so it is allowed to
  replace $F(n)$ by $f(n)$ in our statements.
\end{proof}

\begin{proof}[Proof of Corollary~\ref{cor:jsr-growth}]
  As $F_0$ is the zero-matrix, $F_0^s \in \Oh[\big]{(k^s)^{(\log_k\rho')+\varepsilon}}$ as $s\to\infty$ holds trivially,
  and the result therefore follows from Proposition~\ref{pro:jsr-growth}.
\end{proof}

\begin{proof}[Proof of Theorem~\ref{thm:unsolvable:poly}]
  We reduce from the recursively unsolvable question whether the joint spectral radius of a set of
  matrices over the rationals is bounded by~$1$; see Blondel and
  Tsitsiklis~\cite{Blondel-Tsitsiklis:2000:boundedness-undecidable}.

  So let us assume the decision problem of Theorem~\ref{thm:unsolvable:poly}
  is recursively solvable.
  Let $F_1$, \dots, $F_{k-1}$ be square matrices over the rationals,
  let $F_0$ be the zero-matrix, and
  set $F(n) = F_{n_0} \cdots F_{n_{s-1}}$ for $(n)_k =n_{s-1} \cdots n_0$.
  By the reduction of Lemma~\ref{lem:property:k-regular-products},
  where property~$P$ is
  whether a sequence does not have at least polynomial growth,
  we can decide whether $F(n)$ does not have at least polynomial growth.
  This is equivalent to deciding whether
  for all $\sigma>0$ we have $F(n) \in \Oh{n^\sigma}$.
  By Corollary~\ref{cor:jsr-growth}
  with $\rho'=1$, this is equivalent to deciding whether
  the joint spectral radius $\rho$ of
  $F_1$, \dots, $F_{k-1}$ being at most $1$, a contradiction.
\end{proof}

\section{Images and preimages}

By \cite[Theorem~5.2]{Allouche-Shallit:1992:regular-sequences}, it is
undecidable whether a given $k$-regular sequence $(f(n))_{n \geq 0}$ has a zero
term, i.e., whether there exists an~$n\in\N_0$ with $f(n)=0$.

\subsection{Preimages}

In this section we use closure properties of regular languages without further comment.   See, for example, \cite{Hopcroft&Ullman:1979}.

\begin{theorem}\label{thm:unsolvable:preimage-dfa}
  Let $k\geq2$ be an integer and $\S\supseteq\N_0$ be a set.
  The decision problem
  \begin{quote}
    ``Given a $k$-regular sequence $(f(n))_{n \geq 0}$ over~$\S$ and a number $q\in\S$,
    is the language associated with $f^{-1}(q)$ regular, i.e., can it
    be recognized by a deterministic finite automaton?''
  \end{quote}
  is recursively unsolvable.
\end{theorem}

Above, the language associated with $f^{-1}(q)$ is $\setm{(n)_k}{f(n)=q}$, where
$(n)_k$ is the standard $k$-ary expansion of $n$.

A result from \cite[Theorem~5.3]{Allouche-Shallit:1992:regular-sequences} states
that there exists a $k$-regular sequence $(f(n))_{n\ge0}$ such that neither
$\setm{(n)_k}{f(n)=0}$ nor $\setm{(n)_k}{f(n)\neq0}$ are context-free.

\begin{proof}[Proof of Theorem~\ref{thm:unsolvable:preimage-dfa}]
  We can assume that $q=0$ by subtracting $q$ from the
  $k$-regular sequence. In order to prove the theorem,
  we reduce from Hilbert's tenth problem.

  For a given multivariate polynomial~$p$ in $t$~variables over~$\Z$,
  we choose $r\in\N$ such that $K=k^r > t+2$, and 
  we construct
  \begin{equation}\label{eq:preimage-dfa:f}
    f(n) =  \bigl(\f{p}{\abs{z}_1, \abs{z}_2, \dots, \abs{z}_t}\bigr)^2
    + \bigl(\abs{z}_{t+1} - \abs{z}_{t+2}\bigr)^2,
  \end{equation}
  where $z = (n)_K$.
  The sequence $(f(n))_{n \geq 0}$ is $K$-regular (see Remark~\ref{rem:poly-regular})
  and therefore $k$-regular as well
  by~\cite[Theorem 2.9]{Allouche-Shallit:1992:regular-sequences}.

  The following claim shows that the above indeed provides a reduction.

  \begin{claim*}
    The set $f^{-1}(0)$ is not recognized by a deterministic finite automaton
    iff there exist
    nonnegative integers $x_1$, $x_2$, \ldots, $x_t$ such that
    $p(x_1, \dots, x_t) = 0$.
  \end{claim*}

  If $p(x_1, \dots, x_t) = 0$ has no solution in $\N_0^t$, then $f(n)
  \neq 0$ by its construction. Thus $f^{-1}(0)$ is the empty set and
  is accepted by a deterministic finite automaton.

  Otherwise, suppose we have nonnegative integers $x_1$, \dots, $x_t$ with
  $p(x_1, \dots, x_t) = 0$, and suppose the language $L=\setm{(n)_K}{f(n)=0}$
  is accepted by a deterministic finite automaton, i.e.,
  $L$ is regular.
  We note that each $z=(n)_K \in L$ satisfies $\abs{z}_{t+1}=\abs{z}_{t+2}$
  as $f(n)=0$ and this is equivalent to
  both squares in its definition~\eqref{eq:preimage-dfa:f} being zero.
  Moreover, for each $s\in\N_0$, there is a $z\in L$ with
  $s=\abs{z}_{t+1}=\abs{z}_{t+2}$.

  As $L$ is regular, so is
  \begin{equation*}
    L_1 = L \cap 1^{x_1} 2^{x_2} \cdots t^{x_t} (t+1)^+ (t+2)^+
  \end{equation*}
  where $d^+=\set{d,dd,ddd,\ldots}$ for a letter~$d$.
  Furthermore, the left quotient
  \begin{equation*}
    L_2 =\set{1^{x_1}2^{x_2} \cdots t^{x_t} (t+1)}^{-1} L_1
  \end{equation*}
  is regular. However,
  \begin{equation*}
    L_2 = \setm{(t+1)^{m-1} (t+2)^m}{m\geq1},
  \end{equation*}
  which is not regular;
  see~\cite[Examples 2.8 and 5.2]{Eilenberg:1974:automata-languages-machines-A}.
  This contradiction proves the desired result.

  Now, if we can decide whether the language associated with $f^{-1}(q)$
  is regular, then we can decide whether a solution of $p$ exists, and
  therefore decide Hilbert's tenth problem. This completes the proof of
  Theorem~\ref{thm:unsolvable:preimage-dfa}.
\end{proof}

\subsection{Images}

\begin{theorem}\label{thm:unsolvable:image-N}
  Let $k\geq2$ be an integer and $\S\supseteq\N_0$ be a set.
  The decision problem
  \begin{quote}
    ``Given a $k$-regular sequence $(f(n))_{n \geq 0}$ over~$\S$,
    is $\setm{f(n)}{n \in \N_0} = \N_0$?''
  \end{quote}
  is recursively unsolvable.
\end{theorem}

\begin{proof}
  In order to prove the theorem,
  we again reduce from Hilbert's tenth problem.
  For a given multivariate polynomial~$p$ in $t$~variables over~$\Z$,
  we choose $r\in\N$ such that $K=k^r > t$, and 
  we construct
  \begin{equation*}
    f(n) = 
    \begin{cases}
      n/2 + 1, & \text{if $n$ is even;} \\
      \bigl(\f{p}{\abs{z}_1, \abs{z}_2, \dots, \abs{z}_t}\bigr)^2, &
      \text{if $n$ is odd},
    \end{cases}
  \end{equation*}
  where $z = \bigl((n-1)/2\bigr)_K$.
  The sequence $(f(n))_{n \geq 0}$ is $K$-regular (Remark~\ref{rem:poly-regular}
  and \cite[Theorem~2.7]{Allouche-Shallit:1992:regular-sequences})
  and therefore $k$-regular as well
  by~\cite[Theorem 2.9]{Allouche-Shallit:1992:regular-sequences}.

  Once we have shown the following claim, we have a reduction to
  Hilbert's tenth problem and therefore the proof of
  Theorem~\ref{thm:unsolvable:image-N} is completed.

  \begin{claim*}
    The set $\setm{f(n)}{n \in \N_0}$ equals $\N_0$
    iff there exist
    nonnegative integers $x_1$, $x_2$, \dots, $x_t$ such that
    $p(x_1, \dots, x_t) = 0$.
  \end{claim*}

  If $p(x_1, \dots, x_t) = 0$ has no solution in $\N_0^t$, then $f(n)
  \neq 0$ by its construction. Thus $0$ is not in the set
  $\setm{f(n)}{n \in \N_0}$, so this set cannot be equal to $\N_0$.

  Otherwise, suppose we have nonnegative integers
  $x_1, \ldots, x_t$ with $p(x_1, \dots, x_t) = 0$, then there exists
  an $n\in\N_0$ with $f(n)=0$. As $\setm{f(n)}{\text{$n\in\N_0$ is even}}$
  already contains all the positive integers, all nonnegative integers appear
  as a value $f(n)$ somewhere.
\end{proof}

\begin{theorem}\label{thm:unsolvable:image-Z}
  Let $k\geq2$ be an integer and $\S\supseteq\Z$ be a set.
  The decision problem
  \begin{quote}
    ``Given a $k$-regular sequence $(f(n))_{n \geq 0}$ over~$\S$,
    is $\setm{f(n)}{n \in \N_0} = \Z$?''
  \end{quote}
  is recursively unsolvable.
\end{theorem}

\begin{proof}
  The proof runs along the same lines as the proof of Theorem~\ref{thm:unsolvable:image-N}, but
  for a given multivariate polynomial~$p$ in $t$~variables over~$\Z$,
  we choose $r\in\N$ such that $K=k^r > t$, and 
  we construct
  \begin{equation*}
    f(n) = 
    \begin{cases}
      n/3 + 1, & \text{if $n \equiv 0 \pmod 3$;} \\
      -(n-1)/3 - 1, & \text{if $n \equiv 1 \pmod 3$;} \\
      \f{p}{\abs{z}_1, \abs{z}_2, \dots, \abs{z}_t}, &
      \text{if $n \equiv 2 \pmod 3$,}
    \end{cases}
  \end{equation*}
  where $z = \bigl((n-2)/3\bigr)_K$.
  Then the set $\setm{f(n)}{n \in \N_0}$ equals $\Z$ iff there exist
  nonnegative integers $x_1$, $x_2$, \dots, $x_t$ such that
  $p(x_1, \dots, x_t) = 0$.
\end{proof}

We can extend the above to the question whether two $k$-regular
sequences have the same image.

\begin{corollary}\label{cor:unsolvable:image-Z}
  Let $k\geq2$ be an integer and $\S\supseteq\N_0$ be a set.
  The decision problem
  \begin{quote}
    ``Given two $k$-regular sequences $(f(n))_{n \geq 0}$ and $(g(n))_{n \geq 0}$ over~$\S$,
    do their images coincide, i.e.,
    is $\setm{f(n)}{n \in \N_0} = \setm{g(n)}{n \in \N_0}$?''
  \end{quote}
  is recursively unsolvable.
\end{corollary}

\begin{proof}
  We reduce from the decision problem of
  Theorem~\ref{thm:unsolvable:image-N}, so let $(f(n))_{n \geq 0}$ be a $k$-regular
  sequence over~$\S$ and set~$g(n)=n$. If we can decide whether these
  two sequences have the same image, then we decide whether
  \begin{equation*}
    \setm{f(n)}{n \in \N_0} = \setm{g(n)}{n \in \N_0} = \N_0,
  \end{equation*}
  which contradicts Theorem~\ref{thm:unsolvable:image-N}.
\end{proof}

\begin{theorem}\label{thm:unsolvable:value-twice}
  Let $k\geq2$ be an integer and $\S\supseteq\N_0$ be a set.
  The decision problem
  \begin{quote}
    ``Given a $k$-regular sequence $(f(n))_{n \geq 0}$ over~$\S$,
    does $f(n)$ take the same value twice?''
  \end{quote}
  is recursively unsolvable.
\end{theorem}

\begin{proof}
  In order to prove the theorem,
  we reduce from Hilbert's tenth problem.
  For a given multivariate polynomial~$p$ in $t$~variables over~$\Z$,
  we choose $r\in\N$ such that $K=k^r > t$, and 
  we construct
  \begin{equation*}
    g(m) = 
    \bigl(\f{p}{\abs{z}_1, \abs{z}_2, \dots, \abs{z}_t}\bigr)^2 
  \end{equation*}
  where $z = (m)_K$ and
  \begin{equation*}
    f(n) = \sum_{0 \leq m < n} g(m)
  \end{equation*}
  The sequence $(f(n))_{n \geq 0}$ is $K$-regular (Remark~\ref{rem:poly-regular}
  and \cite[Theorem~3.1]{Allouche-Shallit:1992:regular-sequences})
  and therefore $k$-regular as well
  by~\cite[Theorem 2.9]{Allouche-Shallit:1992:regular-sequences}.

  Once we have shown the following claim, we have a reduction to
  Hilbert's tenth problem and therefore the proof of
  Theorem~\ref{thm:unsolvable:value-twice} is completed.

  \begin{claim*}
    The sequence~$(f(n))_{n \geq 0}$ takes the same value twice
    iff there exist
    nonnegative integers $x_1$, $x_2$, \ldots, $x_t$ such that
    $p(x_1, \dots, x_t) = 0$.
  \end{claim*}

  If $p(x_1, \dots, x_t) = 0$ has no solution in $\N_0^t$, then $g(m)$
  is strictly positive, and therefore $f(n)$ strictly increasing.
  So no value is taken twice.

  Otherwise, suppose we have nonnegative integers
  $x_1$, \dots, $x_t$ with $p(x_1, \dots, x_t) = 0$, then there exists
  an $n\in\N_0$ with $g(n)=0$, and so $f(n) = f(n+1)$.
\end{proof}

\section{Squares and other \texorpdfstring{$\alpha$}{alpha}-powers}

Given a sequence $(f(n))_{n \geq 0}$ and an integer~$\alpha \ge 2$, an 
{\it $\alpha$-power\/} is a nonempty contiguous
subsequence $(f(j))_{i \leq j < i+\alpha m}$ of length $\alpha $,
for some $i$ and $m$, such that
$f(i+t) = f(i+sm+t)$ for all $0 \leq s < \alpha$ and $0 \leq t < m$.
We call a $2$-power a {\it square}.
For example, the fractional part of the 
decimal representation of $e$ contains the
square $18281828$.

A {\it palindrome\/} is a nonempty contiguous subsequence
that reads the same forwards and backwards.  A palindrome is nontrivial if it is of length $\geq 2$.

For automatic sequences, the presence of squares, higher powers, and nontrivial palindromes is decidable
(see, e.g., \cite{Charlier&Rampersad&Shallit:2012}).  We now show
that, in contrast, the existence of these patterns is undecidable for
$k$-regular sequences.

\begin{theorem}\label{thm:undecide-squares}
  Let $\alpha \ge 2$ be an integer. The decision problem
  \begin{quote}
    ``Given a $k$-regular sequence $(f(n))_{n \geq 0}$, does
    $(f(n))_{n \geq 0}$ contain an $\alpha$-power?''
  \end{quote}
  is recursively unsolvable.
\end{theorem}

\begin{proof}
We reduce from the problem of deciding whether a $k$-regular sequence has a $0$ term.

Given a $k$-regular sequence $(f(n))_{n \geq 0}$ for which we want to decide
whether $f(n) = 0$ for some $n$, we can (effectively) transform it to the $k$-regular
sequence $(g(n))_{n \geq 0}$ defined recursively by $g(0)=1$ and
\begin{equation*}
  g(n) = g(n-1) + f(n-(\alpha-1))^2 \cdots f(n-2)^2 \, f(n-1)^2,
  \quad\text{for $n\ge1$}.
\end{equation*}
(Note that we use the convention $f(-i)=1$ for $i\ge1$.)
For squares, this simplifies to the explicit formula
$g(n) = 1 + f(0)^2 + \cdots + f(n-1)^2$.
Then $(g(n))_{n \geq 0}$ is (not necessarily strictly) increasing, so it contains an $\alpha$-power
iff there exists $n \ge 0$ such that $g(n) = g(n+1) = \cdots = g(n+\alpha-1)$.
But this occurs iff $f(n)$ = 0.
\end{proof}

Using the same technique, we can prove following theorem for palindromes.

\begin{theorem}
The decision problem
  \begin{quote}
    ``Given a $k$-regular sequence $(f(n))_{n \geq 0}$, does
    $(f(n))_{n \geq 0}$ contain a nontrivial palindrome?''
  \end{quote}
  is recursively unsolvable.
\end{theorem}

\begin{proof}
The same proof given for squares above works unchanged.
\end{proof}

\renewcommand{\MR}[1]{}
\bibliography{abbrevs,krenn}
\bibliographystyle{amsplainurl}

\end{document}